\documentclass[12pt, a4paper]{article}
\usepackage{hyperref}
\usepackage{amsmath}
\usepackage{amssymb}
\usepackage{amsthm}
\usepackage{fitch}
\usepackage{tabularx}
\usepackage{multicol}
\usepackage{makecell}
\usepackage{thmtools}
\usepackage{enumitem}
\usepackage[bottom]{footmisc}

\newtheoremstyle{break}
  {\topsep}{\topsep}%
  {\itshape}{}%
  {\bfseries}{}%
  {\newline}{}%
\theoremstyle{break}
\newtheorem{definition}{Definition}[section]
\newtheorem{theorem}{Theorem}[section]

\newtheorem{lemma}[theorem]{Lemma}

\author{Yazeed Alrubyli
    \\Department of Mathematics
    \\Politecnico di Milano
    \\Via Edoardo Bonardi, 9
    \\20133 Milano MI, Italy
    \\\href{mailto:yazeednaif.alrubyli@mail.polimi.it}{\texttt{yazeednaif.alrubyli@mail.polimi.it}}}
\date{\today}
\title{Natural Deduction Calculus\\for First-Order Logic}
\begin{document}
    \maketitle
    \paragraph{Abstract.}

    The purpose of this paper is to give an easy to understand with step-by-step explanation to allow interested people to fully appreciate the power of natural deduction for first-order logic. Natural deduction as a proof system can be used to prove various statements in propositional logic, but we will see its extension to cover quantifiers which gives it more power over propositional logic in solving more complex, real-world problems. We started by going over logical connectives and quantifiers to agree on the symbols that will be used throughout the paper, as some authors use different symbols to refer to the same thing. Besides, we showed the inference rules that are used the most. Furthermore, we presented the soundness and completeness of natural deduction for first-order logic. Finally, we solved examples ranging from easy to complex to give you different circumstances in which you can apply the proof system to solve problems you may encounter. Hopefully, this paper will be helpful makes the subject easy to understand.

    \medskip

    
    \newpage
    \tableofcontents


    \newpage
    \section{Introduction}
    First-order logic (a.k.a predict logic) is extending the 
    propositional logic power to give it more power to solve 
    an advanced level of problems that can not be solved with 
    propositional logic. To stay consistent throughout this 
    paper, \textit{predict} as a name will not be used, instead,
    \textit{first-order} is will be used instead \cite{rautenberg2006concise}. 
    By demonstrating with an example the power of first-order 
    logic, let give an example. stating with a proposition P 
    \textit{"Every lion drinks coffee"}, and proposition Q 
    \textit{"Cat is a lion"}, therefore, proposition 
    R \textit{"Cat drinks coffee"}. Using propositional logic, 
    you can not derive the conclusion R from premises P and Q,
    but with first-order logic, you can. From your brain's 
    logical point of view, you can conclude R from P and Q 
    easily as following.

    \medskip

    \begin{fitch}
        P & Every lion drinks milk\\
        Q & Cat is a lion\\
        \therefore R & Therefore, Cat drinks milk 
    \end{fitch}

    \medskip

    \noindent As far as propositional logic is concerned it is hard to 
    arrive at such a conclusion. Propositional logic did not 
    give you the tools to solve such a problem, yet, your 
    brain can logically solve it, also first-order logic can 
    too \cite{russel2013artificial}. Quantifiers from 
    first-order logic such as $\forall$ or $\exists$, i.e.\ 
    universal and existential quantifiers, respectively can 
    help solve such a problem. Let us write the predicates 
    in short notation:

    \begin{itemize}
        \item Lion(x) = x is a lion
        \item Milk(x) = x drinks milk
    \end{itemize}

    \noindent\textbf{Note.} \textit{x} is the subject, \{\textit{is a lion, drinks milk}\}
    are the predicates. Also, the domain of dicourse will 
    be all the animles.

    \medskip

    \begin{fitch}
        \forall{x}(Lion(x)\rightarrow Milk(x)) & premise\\
        Lion(Cat) & premise\\
        Lion(Cat) \rightarrow Milk(Cat) & $\forall$E, 1\\ 
        \therefore Milk(Cat) & $\rightarrow$E, 3, 2
    \end{fitch}

    \medskip
    
    \noindent We discussed ($\forall$E) and ($\rightarrow$E) in the 
    next section \nameref{section:concepts} beside connectives, quantifiers, 
    and other inference rules that will be used to solve such a problem
    using natural deduction for first-order logic.

    \section{Basic Concepts}
    \label{section:concepts}
    Different authors used different symbols, in order to eliminate confuion and stay consistent,
    we will be using the symbols that is widely used in mathematical logic books. Connectives and quantifiers
    are showen in table \ref{tbl:connectives} and table \ref{tbl:quantifiers}, respectivaly.
    \subsection{Logical Connectives}
    The following table is sorted based on the precedence of each connective, starting from the higher
    priorty.
    \def\arraystretch{1.5}

    \begin{table}[h!]
        \centering
        \begin{tabularx}{350pt}{|c|c|X|}
            \cline{1-3}
            Symbol                 & Connective     & Description \\
            \hline
            \hline
            $\top$                 & Truth          & True.                                                                            \\ \cline{1-3}
            $\bot$                 & Falsehood      & False.                                                                           \\ \cline{1-3}
            $\lnot$                & Negation       & $\lnot$P is true if and only if P is false.                                      \\ \cline{1-3}
            $\wedge$               & Conjunction    & P$\wedge$Q is true if and only if both P and Q are true.                         \\ \cline{1-3}
            $\vee$                 & Disjunction    & P$\vee$Q is true if and only if either P or Q is true.                           \\ \cline{1-3}
            $\rightarrow$             & Conditional    & P$\rightarrow$Q is true if and only if either P is false or Q is true (or both).    \\ \cline{1-3}
            $\leftrightarrow $  & Biconditional  & P$\leftrightarrow $Q is true if and only if P and Q have the same truth value.\\ \cline{1-3}
        \end{tabularx}
        \caption{Logiacl Connectives \cite{barker2011language}}
        \label{tbl:connectives}
    \end{table}

    \subsection{Quantifiers}
    \begin{table}[h!]
        \centering
        \begin{tabularx}{350pt}{|c|c|X|}
            \cline{1-3}
            Symbol      & Quantifier    & Description \\
            \hline
            \hline
            $\forall$   & Universal     & $\forall{x}$P, P is true for every object x. \\ \cline{1-3}
            $\exists$   & Existential   & $\exists{x}$P, P is true for at least one object x. \\ \cline{1-3}
        \end{tabularx}
        \caption{Quantifiers \cite{russel2013artificial}}
        \label{tbl:quantifiers}
    \end{table}

    \subsection{Inference Rules}
    \label{section:IR}
    There are many inference rules such as basic, derived, and others. This
    paper covers the basic rules of inference that can be used 
    to derive a proof which is a chain of conclusions that leads
    to the desired outcome \cite{russel2013artificial}.

    \subsubsection{Negation}
    \textit{Negation Introduction} ($\lnot$I) also known as \textit{reductio ad absurdum} is to derive a negation of a sentence
    if the sentence leads to contradiction, As shown on the left, $\bot$ could be Q $\land\lnot$Q
    from assumption P.
    \textit{Negation Elimination} ($\lnot$E) is to reomve the double negation
    \cite{monin2009gentzen}.
    
    \setlength{\columnsep}{-40pt}
    \begin{multicols}{2}
        \begin{center}
            \begin{tabular}{l}
                \textbf{Negation Introduction}\\
                $\phi$\\
                $\bot$\\
                \Xhline{3\arrayrulewidth}
                $\lnot \phi$ 
            \end{tabular}
        \end{center}
        \columnbreak
        \begin{center}
            \begin{tabular}{l}
                \textbf{Negation Elimination}\\
                $\lnot\lnot\phi$\\
                \Xhline{3\arrayrulewidth}
                $\phi$ 
            \end{tabular}
        \end{center}
    \end{multicols}

    \subsubsection{Conjunction}
    \textit{Conjunction Introduction} ($\land$I) is to derive
    a conjunction from its conjuncts, i.e.\ if A is true, and B is true
    , then A$\land$B must be true. 
    \textit{Conjunction Elimination} ($\land$E) is to reomve 
    the conjunction and pick one of its conjuncts, i.e.\ 
    if A$\land$B is true, then A must be true, and B must 
    be true \cite{pfenning2001judgmental}.

    \setlength{\columnsep}{-40pt}
    \begin{multicols}{2}
        \begin{center}
            \begin{tabular}{c}
                \textbf{Conjunction Introduction}\\
                $\phi_1$\\
                $\vdots$\\
                $\phi_n$\\
                \Xhline{3\arrayrulewidth}
                $\phi_1\land...\land\phi_n$
            \end{tabular}
        \end{center}
        \columnbreak
        \begin{center}
            \begin{tabular}{l}
                \textbf{Conjunction Elimination}\\
                $\phi_1\land...\land\phi_n$\\
                \Xhline{3\arrayrulewidth}
                $\phi_i$
            \end{tabular}
        \end{center}
    \end{multicols}

    \pagebreak

    \subsubsection{Disjunction}
    \textit{Disjunction Introduction} ($\lor$I) is to add as many
    disjunct as you prefer if at least one of the disjuncts 
    is in the proof. \textit{Disjunction 
    Elimination} ($\lor$E) is to reomve conjunction and
    pick one of the true sentences \cite{aschieri2016natural}.

    \setlength{\columnsep}{-40pt}
    \begin{multicols}{2}
        \begin{center}
            \begin{tabular}{l}
                \textbf{Disjunction Introduction}\\
                $\phi_i$\\
                \Xhline{3\arrayrulewidth}
                $\phi_1\lor...\lor\phi_n$
            \end{tabular}
        \end{center}
        \columnbreak
        \begin{center}
            \begin{tabular}{c}
                \textbf{Disjunction Elimination}\\
                $\phi_1\lor...\lor\phi_n$\\
                $\phi_1\rightarrow\psi$\\
                $\vdots$\\
                $\phi_n\rightarrow\psi$\\
                \Xhline{3\arrayrulewidth}
                $\psi$
            \end{tabular}
        \end{center}
    \end{multicols}

    \subsubsection{Conditional}
    \textit{Conditional Introduction} ($\rightarrow$I) is to use subproof assuming $\phi$
    and prove $\psi$, then it follows that $\phi\rightarrow\psi$.
    \textit{Conditional Elimination} ($\rightarrow$E) also known as \textit{Modus Ponens}
    is to conclude $\psi$ if $\phi\rightarrow\psi$ and $\phi$ are proven \cite{monin2009gentzen}.

    \setlength{\columnsep}{-40pt}
    \begin{multicols}{2}
        \begin{center}
            \begin{tabular}{l}
                \textbf{Conditional Introduction}\\
                $\phi\vdash\psi$\\
                \Xhline{3\arrayrulewidth}
                $\phi\rightarrow\psi$
            \end{tabular}
        \end{center}
        \columnbreak
        \begin{center}
            \begin{tabular}{l}
                \textbf{Conditional Elimination}\\
                $\phi\rightarrow\psi$\\
                $\phi$\\
                \Xhline{3\arrayrulewidth}
                $\psi$
            \end{tabular}
        \end{center}
    \end{multicols}

    \subsubsection{Biconditional}
    \textit{Biconditional Introduction} ($\leftrightarrow $I) is to use subproof
    by assuming $\phi$ and prove $\psi$, also assuming $\psi$ and prove $\phi$, then it follows
    that $\phi\leftrightarrow\psi$ \cite{barker2011language}. 
    \textit{Biconditional Elimination} ($\leftrightarrow $E) is to replace biconditional by
    $\phi\rightarrow\psi$ or $\psi\rightarrow\phi$ as both of them 
    are true by the definition of biconditional \cite{russel2013artificial}.

    \setlength{\columnsep}{-40pt}
    \begin{multicols}{2}
        \begin{center}
            \begin{tabular}{l}
                \textbf{Biconditional Introduction}\\
                $\phi\rightarrow\psi$\\
                $\psi\rightarrow\phi$\\
                \Xhline{3\arrayrulewidth}
                $\phi\leftrightarrow \psi$
            \end{tabular}
        \end{center}
        \columnbreak
        \begin{center}
            \begin{tabular}{l}
                \textbf{Biconditional Elimination}\\
                $\phi\leftrightarrow \psi$\\
                \Xhline{3\arrayrulewidth}
                $\phi\rightarrow\psi$\\
                $\psi\rightarrow\phi$
            \end{tabular}
        \end{center}
    \end{multicols}

    \pagebreak

    \subsubsection{Universal}
    \textit{Universal Introduction} ($\forall$I) also known as \textit{Universal Generalization}
    if arbtray x has a property $\phi$ i.e $\phi(x)$, then we can conclude that for all x, such that x has propery $\phi$ $\forall{x}\phi(x)$. 
    Variable x should not be free in any hypothesis on which $\phi(x)$ depends.
    \textit{Universal Elimination} ($\forall$E) also known as \textit{Universal Instantiation} 
    if all x in the universe has a property $\phi(x)$, $\forall{x}\phi(x)$, then there must be a \textit{t} in that universe
    that has the property $\phi$, $\phi(t)$, t must be free for x \cite{van2004logic}.

    \setlength{\columnsep}{-40pt}
    \begin{multicols}{2}
        \begin{center}
            \begin{tabular}{c}
                \textbf{Universal Introduction}\\
                $\phi(x)$\\
                \Xhline{3\arrayrulewidth}
                $\forall{x}\phi(x)$
            \end{tabular}
        \end{center}
        \columnbreak
        \begin{center}
            \begin{tabular}{l}
                \textbf{Universal Elimination}\\
                $\forall{x}\phi(x)$\\
                \Xhline{3\arrayrulewidth}
                $\phi(t)$
            \end{tabular}
        \end{center}
    \end{multicols}

    \subsubsection{Existential}
    \textit{Existential Introduction} ($\exists$I) also known as \textit{Existential Generalization}
    if an object c has property P, then there must exist x in a universe that has property P.
    \textit{Existential Elimination} ($\exists$E) also known as \textit{Existential Instantiation}
    if exist in a universe that object x has property P, then there must be an object c that
    has property P \cite{martin1996meanings}.

    \setlength{\columnsep}{-40pt}
    \begin{multicols}{2}
        \begin{center}
            \begin{tabular}{l}
                \textbf{Existential Introduction}\\
                $\phi[t/x]$\\
                \Xhline{3\arrayrulewidth}
                $\exists{x}\phi$
            \end{tabular}
        \end{center}
        \columnbreak
        \begin{center}
            \begin{tabular}{l}
                \textbf{Existential Elimination}\\
                $\exists{x}\phi$\\
                \Xhline{3\arrayrulewidth}
                P(c)
            \end{tabular}
        \end{center}
    \end{multicols}

    \subsubsection{Reiteration}
    As the proof gets complicated and long, iteration rule can be used to bring an earlier step within the proof or to bring it to the subproof. It works as a reminder that "we have already shown that P" \cite{barker2011language}.

    \begin{center}
        \textbf{Reiteration (Re)}\\
        \medskip
        \begin{fitch}
            \fa P & Premise\\
            \fa\fh Q & Assumption\\
            \fa\fa\vdots\\
            \fa\fa P & Re, 1
        \end{fitch}
    \end{center}
    

    \pagebreak

    \section {Natural Deduction Calculus}
    \textit{Natural Deduction} was first introduced as a term by the German logician Gentzen, Gerhard.
    It was introduced as a formalism that mimics how humans naturally reason, hence the name. By applying inference 
    rules, one can infer conclusions from the premises. In other words, it is a method
    for showing that the logical reasoning (premises logically entails conclusion) 
    is valid \cite{ly2017proof}.
    
    \medskip
    
    \textit{Calculus}, which means the way of calculating or reasoning \cite{weiner1993oxford}, in our case, \textit{natural deduction calculus} is calculating the truth values of an argument by means of \textit{natural deduction} proof system. Proofs in natural deduction will be as follows \cite{ly2017proof}:
    \begin{enumerate}
        \item \label{itm:first} Start with zero or more premises.
        \item Prove formula (e.g. P$\land$Q) is provable from \ref{itm:first}.
        \item Justify each formula by using rules of inference or other proper justification. 
    \end{enumerate}

    \medskip

    \noindent In order to use any proof system, we have to make sure it is reliable and correct. Reliability and correctness of a proof system are shown through its \textit{soundenss} and \textit{completeness}.

    \subsection{Soundness}
    The \textit{Soundness Theorem} for a proof system will assure that we can only construct proofs of valid arguments. That is, we want to prove that every sentence in a proof is entailed by the previous sentences \cite{barker2011language}.

    \begin{theorem}[Soundness]\label{thm:soundenss}
        Let $\varphi$ be any formula and $\varGamma$ a set of formulas in first-order language L. If $\varGamma\vdash\varphi$, then $\varGamma\models\varphi$ \cite{goldrei2005propositional}.
        \begin{equation*}
            \varGamma \vdash \varphi \to \varGamma \models \varphi
        \end{equation*}
    \end{theorem}
    
    \begin{proof}
        Using \textit{mathmatical induction}, we can prove the soundness of natural deduction for first-order logic. We start with the base case, i.e. the first step \textbf{n = 1}, if it holds, we do the induction step, where we assume it holds for step \textbf{n = k} and prove that it also holds for the next step \textbf{n = k + 1} \cite{van2004logic}.
        
        \paragraph{}
        \noindent\textbf{Base Case}
        \[\varGamma_1 \models \varphi_1\]
        \[\varphi_1 \models \varphi_1 \] 
        \[{\text{Indeed, any formula is model of itself.}}\]
        \[\therefore \varphi_1 \vdash \varphi_1 \rightarrow \varphi_1 \models \varphi_1\] 
        \paragraph{}
        \textbf{Inductive Step}
        \[Assume \ \varGamma_k \models \varphi_k\]
        \[Show \ \varGamma_{k+1} \models \varphi_{k+1}\]
        Because you are going to use the inference rules in any step of a proof, you need to prove that it will still be valid at any step, in our case, step k+1. Therefore, we will prove the soundness of the proof system case by case. Note that, $\Gamma_{k+1} \models \varphi_{k+1}$ means the valution of $\alpha$ at step k+1 is a model for all formulas in $\Gamma$ at step k+1 and the k previous formulas in $\Gamma$. Furthermore, you can use the truth table instead of the small boxes on the right as proof of what is presented on the left.

        \paragraph{Case 1:}
        Negation Introduction
        \begin{multicols}{2}
            \begin{center}
                \begin{tabular}{l|ll}
                    i   & \multicolumn{1}{||l}{} $\alpha$\\
                    j   & \multicolumn{1}{||r}{} $\bot$\\
                    k+1 & \multicolumn{2}{l}{} $\lnot\alpha$ \ \ \  $\lnot$I, i-j 
                \end{tabular}
            \end{center}
            \columnbreak
            \begin{center}
                \begin{tabular}{|l|}
                    \hline
                    $\Gamma_i \models \alpha$\\
                    $\Gamma_j \models \bot$\\
                    $\Gamma_{k+1} \models \lnot\alpha$ (= $\varphi_{k+1}$)\\
                    \hline
                \end{tabular}
            \end{center}
        \end{multicols}
        \noindent Note that from i to j is a subproof, starting with $\alpha$ as an assumption and reaching a contradiction by the end of the subproof.
        
        \paragraph{Case 2:}
        Negation Elimination
        \begin{multicols}{2}
            \begin{center}
                \begin{tabular}{ll}
                    \multicolumn{1}{l|}{i}   & $\lnot\lnot\alpha$ \\
                    \multicolumn{1}{l|}{k+1} & $\alpha$ \ \ \  $\lnot$E, i
                \end{tabular}
            \end{center}
            \columnbreak
            \begin{center}
                \begin{tabular}{|l|}
                    \hline
                    $\Gamma_i \models \lnot\lnot\alpha$\\
                    $\Gamma_{k+1} \models \alpha$ (= $\varphi_{k+1}$)\\
                    \hline
                \end{tabular}
            \end{center}
        \end{multicols}

        \paragraph{Case 3:}
        Conjunction Introduction
        \setlength{\columnsep}{-40pt}
        \begin{multicols}{2}
            \begin{center}
                \begin{tabular}{ll}
                    \multicolumn{1}{l|}{i}   & $\alpha$ \\
                    \multicolumn{1}{l|}{j}   & $\phi$ \\
                    \multicolumn{1}{l|}{k+1} & $\alpha\land\phi$ \ \ \  $\land$I, i, j
                \end{tabular}
            \end{center}
            \columnbreak
            \begin{center}
                \begin{tabular}{|l|}
                    \hline
                    $\Gamma_i \models \alpha$\\
                    $\Gamma_j \models \phi$\\
                    $\Gamma_{k+1} \models \alpha\land\phi$ (= $\varphi_{k+1}$)\\
                    \hline
                \end{tabular}
            \end{center}
        \end{multicols}

        \vspace{15 pt}

        \noindent\begin{minipage}{\textwidth}
            \paragraph{Case 4:}
            Conjunction Elimination
            \vspace{10 pt}
            \setlength{\columnsep}{-40pt}
            \begin{multicols}{2}
                \begin{center}
                    \begin{tabular}{ll}
                        \multicolumn{1}{l|}{i}   & $\alpha\land\phi$ \\
                        \multicolumn{1}{l|}{k+1} & $\alpha$ \ \ \ \ \ \ \  $\land$E, i\\
                        \hline\\
                        \hline
                    \end{tabular}
                    \begin{tabular}{ll}
                        \multicolumn{1}{l|}{i}   & $\alpha\land\phi$ \\
                        \multicolumn{1}{l|}{k+1} & $\phi$ \ \ \ \ \ \ \  $\land$E, i\\
                    \end{tabular}
                \end{center}
                \columnbreak
                \begin{center}
                    \begin{tabular}{|l|}
                        \hline
                        $\Gamma_i \models \alpha\land\phi$\\
                        $\Gamma_{k+1} \models \alpha$ (= $\varphi_{k+1}$)\\
                        \hline
                    \end{tabular}
                    \newline
                    \vspace*{15 pt}
                    \newline
                    \begin{tabular}{|l|}
                        \hline
                        $\Gamma_i \models \alpha\land\phi$\\
                        $\Gamma_{k+1} \models \phi$ (= $\varphi_{k+1}$)\\
                        \hline
                    \end{tabular}
                \end{center}
            \end{multicols}
        \end{minipage}

        \paragraph{Case 5:}
        Disjunction Introduction
        \setlength{\columnsep}{-40pt}
        \begin{multicols}{2}
            \begin{center}
                \begin{tabular}{ll}
                    \multicolumn{1}{l|}{i}   & $\alpha$ \\
                    \multicolumn{1}{l|}{k+1} & $\alpha\lor\phi$ \ \ \  $\lor$I, i
                \end{tabular}
            \end{center}
            \columnbreak
            \begin{center}
                \begin{tabular}{|l|}
                    \hline
                    $\Gamma_i \models \alpha$\\
                    $\Gamma_{k+1} \models \alpha\lor\phi$ (= $\varphi_{k+1}$)\\
                    \hline
                \end{tabular}
            \end{center}
        \end{multicols}

        \vspace{15 pt}

        \noindent\begin{minipage}{\textwidth}
            \paragraph{Case 6:}
            Disjunction Elimination
            \vspace{10 pt}
            \setlength{\columnsep}{-40pt}
            \begin{multicols}{2}
                \begin{center}
                    \begin{tabular}{ll}
                        \multicolumn{1}{l|}{i}   & $\alpha\lor\phi$ \\
                        \multicolumn{1}{l|}{j}   & $\alpha\to\psi$ \\
                        \multicolumn{1}{l|}{m}   & $\phi\to\psi$ \\
                        \multicolumn{1}{l|}{k+1} & $\psi$ \ \ \  $\lor$E, i, j, m
                    \end{tabular}
                \end{center}
                \columnbreak
                \begin{center}
                    \begin{tabular}{|l|}
                        \hline
                        $\Gamma_i \models \alpha\lor\phi$\\
                        $\Gamma_j \models \alpha\to\psi$\\
                        $\Gamma_m \models \phi\to\psi$\\
                        $\Gamma_{k+1} \models \psi$ (= $\varphi_{k+1}$)\\
                        \hline
                    \end{tabular}
                \end{center}
            \end{multicols}
        \end{minipage}

        \vspace{20 pt}

        \noindent\begin{minipage}{\textwidth}
            \paragraph{Case 7:}
            Conditional Introduction
            \vspace{10 pt}
            \setlength{\columnsep}{-40pt}
            \begin{multicols}{2}
                \begin{center}
                    \begin{tabular}{ll}
                        \multicolumn{1}{l|}{i}   & $\alpha$ \\
                        \multicolumn{1}{l|}{j}   & $\phi$ \\
                        \multicolumn{1}{l|}{k+1} & $\alpha\to\phi$ \ \ \  $\to$I, i, j
                    \end{tabular}
                \end{center}
                \columnbreak
                \begin{center}
                    \begin{tabular}{|l|}
                        \hline
                        $\Gamma_i \models \alpha$\\
                        $\Gamma_j \models \phi$\\
                        $\Gamma_{k+1} \models \alpha\to\phi$ (= $\varphi_{k+1}$)\\
                        \hline
                    \end{tabular}
                \end{center}
            \end{multicols}
        \end{minipage}

        \paragraph{Case 8:}
        Conditional Elimination
        \setlength{\columnsep}{-40pt}
        \begin{multicols}{2}
            \begin{center}
                \begin{tabular}{ll}
                    \multicolumn{1}{l|}{i}   & $\alpha\to\phi$ \\
                    \multicolumn{1}{l|}{j}   & $\alpha$ \\
                    \multicolumn{1}{l|}{k+1} & $\phi$ \ \ \ \ \ \ \   $\to$E, i, j
                \end{tabular}
            \end{center}
            \columnbreak
            \begin{center}
                \begin{tabular}{|l|}
                    \hline
                    $\Gamma_i \models \alpha\to\phi$\\
                    $\Gamma_j \models \alpha$\\
                    $\Gamma_{k+1} \models \phi$ (= $\varphi_{k+1}$)\\
                    \hline
                \end{tabular}
            \end{center}
        \end{multicols}

        \vspace{20 pt}

        \noindent\begin{minipage}{\textwidth}
            \paragraph{Case 9:}
            Biconditional Introduction
            \vspace{10 pt}
            \setlength{\columnsep}{-40pt}
            \begin{multicols}{2}
                \begin{center}
                    \begin{tabular}{ll}
                        \multicolumn{1}{l|}{i}   & $\alpha\to\phi$ \\
                        \multicolumn{1}{l|}{k+1} & $\alpha\leftrightarrow\phi$ \ \ \ \ \  $\leftrightarrow$I, i, j
                    \end{tabular}
                \end{center}
                \columnbreak
                \begin{center}
                    \begin{tabular}{|l|}
                        \hline
                        $\Gamma_i \models \alpha\to\phi$\\
                        $\Gamma_j \models \phi\to\alpha$\\
                        $\Gamma_{k+1} \models \alpha\leftrightarrow\phi$ (= $\varphi_{k+1}$)\\
                        \hline
                    \end{tabular}
                \end{center}
            \end{multicols}
        \end{minipage}

        \vspace{20 pt}
        
        \noindent\begin{minipage}{\textwidth}
            \paragraph{Case 10:}
            Biconditional Elimination
            \vspace{10 pt}
            \setlength{\columnsep}{-40pt}
            \begin{multicols}{2}
                \begin{center}
                    \begin{tabular}{ll}
                        \multicolumn{1}{l|}{i}   & $\alpha\leftrightarrow\phi$ \\
                        \multicolumn{1}{l|}{j}   & $\alpha$ \\
                        \multicolumn{1}{l|}{k+1} & $\phi$ \ \ \ \ \ \ \  $\leftrightarrow$E, i, j\\
                        \hline\\
                        \hline
                    \end{tabular}
                    \begin{tabular}{ll}
                        \multicolumn{1}{l|}{i}   & $\alpha\leftrightarrow\phi$ \\
                        \multicolumn{1}{l|}{j}   & $\phi$ \\
                        \multicolumn{1}{l|}{k+1} & $\alpha$ \ \ \ \ \ \ \  $\leftrightarrow$E, i, j\\
                    \end{tabular}
                \end{center}
                \columnbreak
                \begin{center}
                    \begin{tabular}{|l|}
                        \hline
                        $\Gamma_i \models \alpha\leftrightarrow\phi$\\
                        $\Gamma_j \models \alpha$\\
                        $\Gamma_{k+1} \models \phi$ (= $\varphi_{k+1}$)\\
                        \hline
                    \end{tabular}
                    \newline
                    \vspace*{15 pt}
                    \newline
                    \begin{tabular}{|l|}
                        \hline
                        $\Gamma_i \models \alpha\leftrightarrow\phi$\\
                        $\Gamma_j \models \phi$\\
                        $\Gamma_{k+1} \models \alpha$ (= $\varphi_{k+1}$)\\
                        \hline
                    \end{tabular}
                \end{center}
            \end{multicols}
        \end{minipage}
        
        \paragraph{Case 11:}
        Universal Introduction
        \begin{multicols}{2}
            \begin{center}
                \begin{tabular}{l|ll}
                    i   & \multicolumn{1}{||l}{} $c$\\
                    j   & \multicolumn{1}{||r}{} $\phi(c)$\\
                    k+1 & \multicolumn{2}{l}{} $\forall{x}\phi(x)$ \ \ \  $\forall$I, i-j
                \end{tabular}
            \end{center}
            \columnbreak
            \begin{center}
                \begin{tabular}{|l|}
                    \hline
                    $\Gamma_i \models c$\\
                    $\Gamma_j \models \phi(c)$\\
                    $\Gamma_{k+1} \models \forall{x}\phi(x)$ (= $\varphi_{k+1}$)\\
                    \hline
                \end{tabular}
            \end{center}
        \end{multicols}
        
        \noindent Note that c is an arbitrary object from the domain of discourse that must be introduced as a new constant in a subproof, then prove that c has a property $\phi$, i.e. $\phi(c)$. $\phi(c)$ must not contain any constant introduced by existential elimination after we introduced the constant c \cite{barker2011language}.
        
        \paragraph{Case 12:}
        Universal Elimination
        \begin{multicols}{2}
            \begin{center}
                \begin{tabular}{ll}
                    \multicolumn{1}{l|}{i}   & $\forall{x}\phi(x)$ \\
                    \multicolumn{1}{l|}{k+1} & $\phi(c)$ \ \ \ \ \  $\forall$E, i
                \end{tabular}
            \end{center}
            \columnbreak
            \begin{center}
                \begin{tabular}{|l|}
                    \hline
                    $\Gamma_i \models \forall{x}\phi(x)$\\
                    $\Gamma_{k+1} \models \phi(c)$ (= $\varphi_{k+1}$)\\
                    \hline
                \end{tabular}
            \end{center}
        \end{multicols}

        \vspace{20 pt}
        
        \noindent\begin{minipage}{\textwidth}
            \paragraph{Case 13:}
            Existential Introduction
            \vspace{10 pt}
            \begin{multicols}{2}
                \begin{center}
                    \begin{tabular}{ll}
                        \multicolumn{1}{l|}{i}   & $\phi(c)$ \\
                        \multicolumn{1}{l|}{k+1} & $\exists{x}\phi(x)$ \ \ \ \ \  $\exists$E, i
                    \end{tabular}
                \end{center}
                \columnbreak
                \begin{center}
                    \begin{tabular}{|l|}
                        \hline
                        $\Gamma_i \models \phi(c)$\\
                        $\Gamma_{k+1} \models \exists{x}\phi(x)$ (= $\varphi_{k+1}$)\\
                        \hline
                    \end{tabular}
                \end{center}
            \end{multicols}
        \end{minipage}

        \vspace{20 pt}
        
        \noindent\begin{minipage}{\textwidth}
            \paragraph{Case 14:}
            Existential Elimination
            \vspace{10 pt}
            \begin{multicols}{2}
                \begin{center}
                    \begin{tabular}{ll}
                        \multicolumn{1}{l|}{i}   & $\exists{x}\phi(x)$ \\
                        \multicolumn{1}{l|}{k+1} & $\phi(c)$ \ \ \ \ \  $\exists$E, i
                    \end{tabular}
                \end{center}
                \columnbreak
                \begin{center}
                    \begin{tabular}{|l|}
                        \hline
                        $\Gamma_i \models \exists{x}\phi(x)$\\
                        $\Gamma_{k+1} \models \phi(c)$ (= $\varphi_{k+1}$)\\
                        \hline
                    \end{tabular}
                \end{center}
            \end{multicols}
            \noindent Note that c is an object that satisfies property $\phi$. Therefore, you may assume $\phi(c)$ \cite{barker2011language}.
        \end{minipage}
        
        \paragraph{Case 15:}
        Reiteration
        
        \setlength{\columnsep}{-40pt}
        \begin{multicols}{2}
            \begin{center}
                \begin{tabular}{ll}
                    \multicolumn{1}{l|}{i}   & $\alpha$ \\
                    \multicolumn{1}{l|}{k+1} & $\alpha$ \ \ \  Reit, i
                \end{tabular}
            \end{center}
            \columnbreak
            \begin{center}
                \begin{tabular}{|l|}
                    \hline
                    $\Gamma_i \models \alpha_h$\\
                    $\Gamma_{k+1} \models \alpha_{k+1}$\\
                    \hline
                \end{tabular}
            \end{center}
        \end{multicols}
        \noindent From the 15 cases that have been shown, we can conclude that the natural deduction for first-order logic is sound. \qedhere
    \end{proof}

    \subsection{Completeness}
    Gödel's completeness theorem states that a deduction system is said to be complete if every universally valid formula in the language L has a proof under the proof system, (natural deduction) in our case \cite{koepke2007godel}.

    \medskip

   Proving the completeness of a formal proof system is a huge and complex task, it was Gödel's doctoral dissertation that was finished in 1929 and published in 1930 \cite{baaz2011kurt}. Therefore, to stay consistent in the way we present the proof to fit with the overall presentation of the paper in which we aim to make various
   concepts as easy to understand as possible, we will give a sketch of the proof, but before doing so, we need to stop at a couple of definitions and lemmas on our way to reach the proof of completeness for the natural deduction for first-order logic \footnote{The procedure we followed to tackle the proof is taken from the book \textit{Logic and Structure by Dirk van Dalen} \cite{van2004logic}.}.

   \vspace{20 pt}
        
   \noindent\begin{minipage}{\textwidth}
        \begin{definition}
        (i) A theory T is a collection of sentences with property $T\vdash\varphi\to\varphi\in T$ (T is closed under derivability).\\
        (ii) A set $\Gamma$ such that T = $\{\varphi \ | \ \Gamma\vdash\varphi\}$ is called an axiom set of theory T.\\
        (iii) Theory T is called Henkin theory, if for each sentence $\exists{x}\varphi(x)$ there is a constant c such that $\exists{x}\varphi(x)\to\varphi(c)\in T$ (c is called a witness for $\exists{x}\varphi(x)$).
        \end{definition}
    \end{minipage}
    
    \vspace{2 pt}

    \begin{definition}
        Let T, T' be theories in language L, L'.\\
        (i) T' is an extension of T if $T\subseteq T'$.\\
        (ii) T' is a conservative extension of T if $T'\cap L=T$, i.e. all theorems of T' in the language L are already theorems of T.
    \end{definition}

    \begin{definition}
        Let a theory T be with language L. By adding a constant $c_\varphi$ for each sentence of the form $\exists{x}\varphi(x)$in language L, we obtain L*. T* is the theory with axiom set $T\cup\{\exists{x}\varphi(x)\to\varphi(c) \ | \ \exists{x}\varphi(x)$ closed, with witness $c_\varphi\}$
    \end{definition}

    \begin{lemma}
        Let language L have cardinality $\kappa$ . If $\Gamma$ is a consistent set of sentences, then $\Gamma$ has a model of cardinality $\leq \kappa$.

    \end{lemma}

    \begin{lemma}
        T* is conservative over T.
    \end{lemma}
    \begin{proof}
        \begin{enumerate}[label=(\alph*)]
            \item \label{itm:a} Let $\exists{x}\alpha(x)\to\alpha(c)$ be one of the new axioms.\\
        Suppose set of sentences $\Gamma,\exists{x}\alpha(x)\to\alpha(c)\vdash\psi$, where the constant \textit{c} is neither in $\Gamma$ nor in $\psi$. We will show that $\Gamma\vdash\psi$:
            \begin{enumerate}[label=(\arabic*)]
                \item $\Gamma\vdash(\exists{x}\alpha(x)\to\alpha(c))\to\psi$.
                \item $\Gamma\vdash(\exists{x}\alpha(x)\to\alpha(y))\to\psi$. Note that \textit{y} is a varibale that does not occure in the associated derivation. 2 follows from 1, it harmless to replace \textit{c} by \textit{y}, the derivation remains intact).
                \item $\Gamma\vdash\forall{y}[(\exists{x}\alpha(x)\to\alpha(y))\to\psi]$. Since \textit{c} does not occure in $\Gamma$, the application of $\forall$ is valid.
                \item $\Gamma\vdash\exists{y}(\exists{x}\alpha(x)\to\alpha(y))\to\psi$.
                \item \label{itm:f} $\Gamma\vdash(\exists{x}\alpha(x)\to\exists{y}\alpha(y))\to\psi$.
                \item \label{itm:s} $\vdash\exists{x}\alpha(x)\to\exists{y}\alpha(y)$.
                \item From \ref{itm:f}, \ref{itm:s}, $\Gamma\vdash\psi$.
            \end{enumerate}
            \item Let \textit{T*} $\vdash\psi$, we know that $T \cup \{\delta_1,...,\delta_n\}\vdash\psi$ from derivability's definition, where $\delta_i$ is the new axiom of the form $\exists{x}\alpha(x)\to\alpha(c)$. We will prove $T\vdash\psi$ by induction. For the base case, where \textit{n = 0}, is done by \ref{itm:a}. For inductive step, let $T \cup \{\delta_1,...,\delta_n\}\vdash\psi$. Set $T' = T \cup \{\delta_1,...,\delta_n\}\vdash\psi$, then $T', \delta_{n+1}\vdash\psi$. By induction hypothesis, $T\vdash\psi$.{\qedhere}
        \end{enumerate}
    \end{proof}

    \begin{lemma}
        Define $T_0 := T, T_{n+1} := (T_n)$* $T_\omega := \cup \{T_n \ | \ n \geq  0 \}$. Then $T_\omega$ is a Henkin theory and it is conservative over T.
    \end{lemma}
    \begin{proof}
        Call $L_n$ the language of $T_n$ and $L_\omega$ the language of $T_\omega$.
        \begin{enumerate}[label=(\roman*)]
            \item $T_n$ is conservative over T.
            \item $T_\omega$ is a theory. Suppose $T_\omega\vdash\delta$, then $\alpha_0,...,\alpha_n\in T_\omega$. For each $i \le n, \alpha_i\in T_{m_i}$ for some $m_i$. Let $m = max\{m_i|i \le n\}$. $T_{m_i}\subseteq T_m(i \le n)$ since for all k, $T_k \subseteq T_{k+1}$. Therefore, $T_m \vdash\delta$. $T_m$ is a theory by definition, so $\delta\in T_m\subseteq L_\omega$
            \item $T_\omega$ is a Henkin theory. Let $\exists{x}\alpha(x)\in L_\omega$, then $\exists{x}\alpha(x)\in L_n$. $\exists{x}\alpha(x)\to\alpha(c)\in L_{n+1}$ (by definition) for a certain c. So, $\exists{x}\alpha(x)\to\alpha(c)\in L_\omega$.
            \item $T_\omega$ is conservative over T. Note that $T_\omega\vdash\delta$ if $T_n\vdash\delta$ for some $n$.{\qedhere}
        \end{enumerate}
    \end{proof}

    \begin{lemma}[Lindenbaum]
        Each consistent theory is contained in a maximally consistent theory.
    \end{lemma}
    \begin{proof}
        Let $T$ be consistent. Consider the partially ordered by inclusion set $A$ of all consistent extensions $T'$ of $T$. We claim that $A$ has a maximal element.
        \begin{enumerate}
            \item \label{itm:1} All chains in $A$ has an upper bound. Let $\{T_i|i \in I\}$ be a chain, then $T' = \bigcup T_i$ is a consistent extension of $T$ containing each $T_i$. So $T'$ is an upper bound.
            \item From \ref{itm:1}, $A$ has a maximal element $T_m$ (Zorn’s lemma).
            \item Trivially we can see that $T_m$ is a maximally consistent extension of $T$, in the sense of $\subseteq$, therefore, $T$ is contained in the maximally consistent theory $T_m$.{\qedhere}
        \end{enumerate}
    \end{proof}
    \newpage
    \begin{lemma}
        An extension of a Henkin theory with the same language is again a Henkin theory.
    \end{lemma}

    \begin{lemma}[Model Existence Lemma]\label{lem:mel}
        If $\Gamma$ is consistent, then $\Gamma$ has a model.
    \end{lemma}
    \begin{proof}
        Let the theory given by $T$ to be $T=\{\delta|\Gamma\vdash\delta\}$. Trivially, any model of $T$ is also amodel of $\Gamma$. Let the maximally consistent Henkin extension of $T$ to be $T_m$ using $T_m$ itself. Recall, that the language is nothing but a set of strings of symbols.
        \begin{enumerate}
            \item $A=\{t \in L_m| t $ is closed$\}$.
            \item We define a function $\hat{f}(t_1,...,t_k):=f(t_1,...,t_k)$, for each function symbol $\bar{f}$.
            \item We define a relation $\hat{P}\subseteq A^p$ by $<t_1,...,t_p> \in \hat{P} \leftrightarrow T_m \vdash P(t_1,...,t_p)$, for each predicate symbol $\bar{P}$.
            \item We define a constant $\hat{c}:=c$, for each constant $c$.
        \end{enumerate}
        We can assert that:
        \begin{enumerate}[label=(\alph*)]
            \item The relation $t \sim s$ defined by $T_m \vdash t = s for t, s \in A$ is an equivalence relation.
            \item $t_i \sim s_i \ (i \le p)$ and $<t_1,...,t_p> \in \hat{P}\to <s_1,...,s_p>\in \hat{P}$. $t_i \sim s_i \ (i \le k) \to \hat{f}(t_1,...,t_k) \sim \hat{f}(s_1,...,s_k)$ for all symbols $P$ and $f$.
        \end{enumerate}
        As we have the equivalence relation, it is natural to introduce the quotient structure.\\
        Denote the equivalence class of $t$ under $\sim$ by $[t]$.\\
        Define $\mathfrak{A} := <A/\sim,\tilde{P_1},...,\tilde{P_n}, \tilde{f_1},...,\tilde{f_m},\{\tilde{c_i}|i\in I\}$, where:
        \begin{itemize}
            \item $\tilde{P_i}:=\{<[t_1],...,[t_{r_i}]> | <[t_1],...,[t_{r_i}]> \in \hat{P_i}\}$.
            \item $\tilde{f_j}([t_1],...,[t_{a_j}])=[\hat{f_j}(t_1,...,t_{a_j})]$.
            \item $\tilde{c_i}:=[\hat{c_i}]$.
        \end{itemize}
        \newpage
        \noindent By induction we can prove $\mathfrak{A}\models\alpha(t)\leftrightarrow T_m \vdash \alpha(t)$ for all sentences in the language $L_m$ of $T_m$ (a.k.a $L(\mathfrak{A})$)
        \begin{enumerate}[label=(\roman*)]
            \item $\alpha$ is atomic. $\mathfrak{A}\models P(t_1,...,t_p)\leftrightarrow <t^{\mathfrak{A}}_1,...,t^{\mathfrak{A}}_p>\in\tilde{p}\leftrightarrow<[t_1],...,[t_p]>\in\tilde{P}\leftrightarrow<t_1,...,t_p>\in\hat{P}\leftrightarrow T_m\vdash P(t_1,...,t_p)$.
            \item Trivially, $\alpha = \delta \land \tau$.
            \item $\alpha = \delta \to \tau$. We can see that $T_m\vdash\delta\to\tau \leftrightarrow(T_m\vdash\delta\to T_m\vdash\tau)$.
            \item $\alpha = \forall{x}\psi(x)$. $\mathfrak{A}\models\forall{x}\psi(x)\leftrightarrow\mathfrak{A}\nvDash\exists{x}\lnot\psi(x)\leftrightarrow\mathfrak{A}\nvDash\lnot\psi(\bar{a})$, for all $a \in|\mathfrak{A}|\leftrightarrow$ for all $a \in |\mathfrak{A}|(\mathfrak{A}\models\psi(\bar{a}))$. We assume $\mathfrak{A}\models\forall{x}\psi(x)$, we get $\mathfrak{A}\models\psi(c)$ for witness c belong to $\exists{x}\lnot\psi(x)$. By induction hypothesis $T_m \vdash\psi(c)$. $T_m\vdash\exists{x}\lnot\psi(x)\to\lnot\psi(c)$, so $T_m \vdash\psi(c)\to\lnot\exists{x}\lnot\psi(x)$. Thus, $T_m\vdash\forall{x}\alpha(x)$.\\
            Contrarily, $T_m\vdash\forall{x}\psi(x)\to T_m\vdash\psi(t)$, so $T_m \vdash \psi(t)$ for all closed $t$. By induction hypothesis, $\mathfrak{A}\models\psi(t)$ for all closed $t$. Thus, $\mathfrak{A} \models\forall{x}\psi(x)$. We can see that $\mathfrak{A}$ is a model of $\Gamma$, as $\Gamma\subseteq T_m$. {\qedhere} 
        \end{enumerate}
        The model constructed above is known canonical model or the closed term model.
        
    \end{proof}

    From \ref{lem:mel} we can immediately deduce Gödel's completeness theorem

    \begin{theorem}[Completeness]\label{thm:completeness}
        Let $\varphi$ be any formula and $\varGamma$ a set of formulas in first-order language L. If $\varGamma \models \varphi$, then $\varGamma \vdash \varphi$ \cite{goldrei2005propositional}.
            \begin{equation*}
                \varGamma \models \varphi \to \varGamma \vdash \varphi
            \end{equation*}
    \end{theorem}

    \pagebreak


    \section{Examples}
    There are different styles for representing the proof of an argument, e.g. Gentzen-style, Fitch-style, and others. This paper will follow Fitch-style to solve the examples in this section \cite{ly2017proof}.
    
    \subsection{Mortality and Socrates}
    All humans are mortal, Socrates is human. Therefore, someone
    is mortal \cite{barker2011language}.
    \begin{itemize}
        \item H(x): x is human
        \item M(x): x is mortal
        \item s: Socrates
    \end{itemize}

    \medskip

    \noindent\textbf{Indirect Proof (Proof by Contradiction)}\\
    \begin{fitch}
        \fa\forall{x}(H(x)\rightarrow M(x))    & Premise\\
        \fa H(s)                                & Premise\\
        \fa\fh\lnot\exists{x}M(x)                  & Assumption\\ 
        \fa\fa\forall{x}\lnot M(x)                 & Def, 3\\
        \fa\fa H(s)\rightarrow M(s)                & $\forall$E, 1\\
        \fa\fa\lnot M(s)                           & $\forall$E, 4\\
        \fa\fa M(s)                                & $\rightarrow$E, 5, 2\\
        \fa\fa\bot                                 & $\bot$, 6, 7\\
        \fa\exists{x}M(x)                      & $\lnot$I, 3-8
    \end{fitch}

    \medskip

    \noindent Start with assuming $\lnot\exists{x}$M(x) and try to find a counter-example, hence the name (proof by contradiction). Reductio ad absurdum is a rule to show that if an assumption leads to a contradiction, then the negation of that assumption must be true \cite{van2004logic}.

    \medskip

    \noindent\textbf{Direct Proof}\\
    \begin{fitch}
        \forall{x}(H(x)\rightarrow M(x))    & Premise\\
        H(s)                                & Premise\\
        H(s) \rightarrow M(s)               & $\forall$E, 1\\ 
        M(s)                                & $\rightarrow$E, 3, 2\\
        \exists{x}M(x)                      & $\exists$I, 4
    \end{fitch}

    \pagebreak

    \subsection{Living Trees}
    All trees are plants, All plants are living things. Therefore, all
    trees are living things.
    \begin{itemize}
        \item T(x): x is tree
        \item P(x): x is plant
        \item L(x): x is a living thing
    \end{itemize}

    \medskip

    \noindent\textbf{Indirect Proof (Proof by Contradiction)}\\
    \begin{fitch}
        \fa\forall{x}(T(x)\rightarrow P(x))    & Premise\\
        \fa\forall{x}(P(x)\rightarrow L(x))    & Premise\\
        \fa\fh\lnot\forall{x}(T(x)\rightarrow L(x))     & Assumption\\ 
        \fa\fa\exists{x}\lnot(T(x)\rightarrow L(x))     & Def, 3\\
        \fa\fa\lnot(T(a)\rightarrow L(a))               & $\exists$E, 4\\
        \fa\fa T(a) \land\lnot L(a)                     & EQUIV, 5\\
        \fa\fa T(a)\rightarrow P(a)                     & $\forall$E, 1\\
        \fa\fa P(a)\rightarrow L(a)                     & $\forall$E, 2\\
        \fa\fa T(a)                                     & $\land$E, 6\\
        \fa\fa P(a)                                     & $\rightarrow$E, 7, 9\\
        \fa\fa \lnot L(a)                               & $\land$E, 6\\
        \fa\fa L(a)                                     & $\rightarrow$E, 8, 10\\
        \fa\fa\bot                                      & $\bot$, 11, 12\\
        \fa\forall{x}(T(x)\rightarrow L(x))    & $\lnot$I, 3-13
    \end{fitch}

    \pagebreak

    \noindent\textbf{Direct Proof}\\
    \begin{fitch}
        \fa\forall{x}(T(x)\rightarrow P(x))    & Premise\\
        \fa\forall{x}(P(x)\rightarrow L(x))    & Premise\\
        \fa T(a) \rightarrow P(a)              & $\forall$E, 1\\ 
        \fa P(a) \rightarrow L(a)              & $\forall$E, 2\\
        \fa\fh T(a)                             & Assumption\\
        \fa\fa P(a)                             & $\rightarrow$E, 3, 5\\
        \fa\fa L(a)                             & $\rightarrow$E, 4, 6\\ 
        \fa T(a) \rightarrow L(a)              & $\rightarrow$I, 5-7\\ 
        \fa\forall{x}(T(x)\rightarrow L(x))    & $\forall$I, 8
    \end{fitch}

    \subsection{Cats and Rabitts}
    Some cats have fur or some cat are rabbits. Therefore,
    some cats are rabbits or have fur.
    \begin{itemize}
        \item F(x): x has fur
        \item R(x): x is a rabbit
    \end{itemize}

    \noindent\textbf{Indirect Proof (Proof by Contradiction)}\\
    \begin{fitch}
        \fa\exists{x}F(x)\lor\exists{x}R(x)      & Premise\\
        \fa\fh\lnot\exists{x}(F(x) \lor R(x))       & Assumption\\
        \fa\fa F(c)                                 & Assumption\\                     
        \fa\fa\forall{x}\lnot(F(x) \lor R(x))       & Def, 2\\
        \fa\fa\lnot F(c)                            & $\forall$E, 5\\
        \fa\fa\bot                                  & $\bot$, 3, 6\\
        \fa\exists{x}(F(x) \lor R(x))            & $\lnot$I, 2-6
    \end{fitch}

    \pagebreak

    \noindent\textbf{Direct Proof}\\
    \begin{fitch}
        \fa\exists{x}F(x)\lor\exists{x}R(x)      & Premise\\
        \fa\fh\exists{x}F(x)                     & Assumption\\
        \fa\fa F(c)                              & $\exists$E, 2\\
        \fa\fa F(c)\lor R(c)                     & $\lor$I, 3\\
        \fa\exists{x}(F(x) \lor R(x))            & $\exists$I, 4\\
    \end{fitch}

    \newpage
    \bibliographystyle{plain}
    \bibliography{biblio.bib}
\end{document}